\documentclass{article}

\usepackage{hyperref} 
\usepackage{arxiv}

\usepackage{amsmath}
\usepackage{amssymb}
\DeclareMathOperator{\Laplace}{\mathcal{L}}
\DeclareMathOperator{\vpa}{vpa}
\DeclareMathOperator{\double}{double}
\DeclareMathOperator{\find}{find}
\DeclareMathOperator{\abs}{abs}
\usepackage{amsthm}
\usepackage{graphicx}
\graphicspath{{Figures/}}
\usepackage{caption}
\usepackage{subcaption}
\graphicspath{{Figures/}}
\usepackage{adjustbox}
\usepackage{cases}
\usepackage{xcolor}

\usepackage{algorithm}
\usepackage{algpseudocode}
\makeatletter
\newenvironment{breakablealgorithm}
  {
   \begin{center}
     \refstepcounter{algorithm}
     \hrule height.8pt depth0pt \kern2pt
     \renewcommand{\caption}[2][\relax]{
       {\raggedright\textbf{\ALG@name~\thealgorithm} ##2\par}%
       \ifx\relax##1\relax 
         \addcontentsline{loa}{algorithm}{\protect\numberline{\thealgorithm}##2}%
       \else 
         \addcontentsline{loa}{algorithm}{\protect\numberline{\thealgorithm}##1}%
       \fi
       \kern2pt\hrule\kern2pt
     }
  }{
     \kern2pt\hrule\relax
   \end{center}
  }
\makeatother

\usepackage{circuitikz}
\usepackage{tikz}
\usepackage{changes}
\definechangesauthor[name={jjuchem}, color=red]{jj}

\usepackage{amsthm}
\newtheorem{theo}{Theorem}[section]
\newtheorem{lem}{Lemma}[section]

\newtheorem{cor}{Corollary}[section]

\title{First Order Plus Fractional Diffusive Delay Modeling: interconnected discrete systems}
\author{
	Jasper Juchem \\
	Department of Electromechanical, Systems and Metal Engineering\\
	Ghent University\\
	9052 Zwijnaarde, Belgium \\
	\texttt{Jasper.Juchem@UGent.be}\\
\And
	Amélie Chevalier \\
	Department of Electromechanical, Systems and Metal Engineering\\
	Ghent University\\
	9052 Zwijnaarde, Belgium \\
	\texttt{Amelie.Chevalier@UGent.be}\\
\And
	Kevin Dekemele \\
	Department of Electromechanical, Systems and Metal Engineering\\
	Ghent University\\
	9052 Zwijnaarde, Belgium \\
	\texttt{Kevin.Dekemele@UGent.be}\\
\And
	Mia Loccufier \\
	Department of Electromechanical, Systems and Metal Engineering\\
	Ghent University\\
	9052 Zwijnaarde, Belgium \\
	\texttt{Mia.Loccufier@UGent.be}\\
}

\begin{document}
	
\maketitle

\begin{abstract}
This paper presents a novel First Order Plus Fractional Diffusive Delay (FOPFDD) model, capable of modeling delay dominant systems with high accuracy. The novelty of the FOPFDD is the Fractional Diffusive Delay (FDD) term, an exponential delay of non-integer order $\alpha$, i.e. $e^{-(Ls)^{\alpha}}$ in Laplace domain. The special cases of $\alpha = 0.5$ and $\alpha = 1$ have already been investigated thoroughly. In this work $\alpha$ is generalized to any real number in the interval $]0,1[$. For $\alpha=0.5$, this term appears in the solution of distributed diffusion systems, which will serve as a source of inspiration for this work. Both frequency and time domain are investigated. However, regarding the latter, no closed-form expression of the inverse Laplace transform of the FDD can be found for all $\alpha$, so numerical tools are used to obtain an impulse response of the FDD. To establish the algorithm, several properties of the FDD term have been proven: firstly, existence of the term, secondly, invariance of the time integral of the impulse response, and thirdly, dependency of the impulse response's energy on $\alpha$. To conclude, the FOPFDD model is fitted to several delay-dominant, diffusive-like resistors-capacitors (RC) circuits to show the increased modeling accuracy compared to other state-of-the-art models found in literature. The FOPFDD model outperforms the other approximation models in accuratly tracking frequency response functions as well as in mimicing the peculiar delay/diffusive-like time responses, coming from the interconnection of a large number of discrete subsystems. The fractional character of the FOPFDD makes it an ideal candidate for an approximate model to these large and complex systems with only a few parameters.
 \medskip



\end{abstract}


\section{Introduction}\label{sec:Intro}
\setcounter{section}{1}
\setcounter{equation}{0}

	Modeling of complex and interconnected processes has been a focus of research for many years \cite{Swaroop1996,Lu2019}. Obtaining a system's model provides insight into the dynamical behavior and stability of a process and allows for the design of controllers. Finding highly accurate models of complex systems is often time-consuming, expensive and superfluous.
	Therefore, literature reports on an approximation of higher order systems with First Order Plus Dead Time (FOPDT) systems \cite{Sudaresan}. The FOPDT model is simple and easy to fit. However, the low complexity of this approximation often fails to model the complete dynamical behavior in both time and frequency domain. A possible solution is higher order fitting, such as Second Order Plus Dead Time (SOPDT) models. However, this often only yields small improvements over an increased complexity \cite{Chang, Balaguer}. 

	Fractional order calculus has been used to improve model fitting in many fields of engineering such as bio-engineering \cite{Magin}, control engineering \cite{Monje2010,Caponetto}, and electrical engineering \cite{Sierociuk} to name a few. These new models rely on the theory of integration and differentiation of an arbitrary real order which is not necessarily integer \cite{Podlubny1999,Oldham}. The fractional order, denoted by $\alpha$, can be any real or complex number. The integer order case, where $\alpha\in\mathbb{N}$, is a particular case of the more general fractional order calculus.

	Recently, research is focused on integrating fractional order calculus into the low complexity FOPDT model to obtain better model fits, yielding to the Fractional Order First Order Plus Dead Time (FO$^2$PDT) model \cite{Narang2011}. Other enhancements to the FOPDT model are reported in literature \cite{Wang,Kaya}. Srinivasan and Chidambaram used the Laplace transform approach and the modified relay feedback method to improve the FOPDT method by introducing an extra modeling parameter \cite{Srinivasan}.

	The current research has been inspired by two observations: i) the solution of diffusion equations, a class of partial differential equations (PDE), contains a fractional order exponential term $e^{-(Ls)^{0.5}}$ in Laplace domain; and ii) in \cite{Sierociuk2015}, thermal diffusion processes have been modeled as lumped-parameter resistors-capacitors (RC) networks. This work combines both observations into the following hypothesis: ``Using the generalized term $e^{-(Ls)^{\alpha}}$, delay dominant systems, which can be expressed as a finite series of interconnected discrete subsystems, are modeled with an increased accuracy." This fractional exponential term will be referred to as the Fractional Diffusive Delay (FDD) term. Examples of these delay-dominant systems, also known as process reaction curves, are electrical circuits, smart grids \cite{Zou,Kanchev} and municipal water systems \cite{Ocampo}. 
	
	This new FDD term opens the door to a new fractional approximation model: the First Order Plus Fractional Diffusive Delay (FOPFDD) model. An exploratory research \cite{Juchem2019} first investigated the combination of FDD and a first-order model in frequency domain. The findings of Juchem et al. (2019) were later used in \cite{Muresan2020} to examine stability margins related to closed-loop behavior. The time domain aspect and the relation to time delay and diffusion are not explored in these preliminary works.
	
	The current work gives more theoretical background on FDD, including an extensive discussion in time domain, which is never done before. The time domain response, following from the inverse Laplace transform of $e^{-(Ls)^{\alpha}}$, allows us to better understand the effect of the FDD's parameters. Preliminary results regarding the link with diffusion and delay are given as well. Furthermore, the performance of the new model is compared to the other models mentioned before for an RC circuit of variable size.

	This paper is structured as follows: the next section provides a link between the FDD and the diffusive partial differential equation both in frequency and time domain. Numerical simulation issues are addressed as well. The third section presents the FOPFDD modeling method. The fourth section presents the performance of the approximation model while the fifth section provides a discussion. The final section gives a main conclusion and future work. 

\section{Origins \& Theoretical Concepts}
\setcounter{section}{2}
\setcounter{equation}{0}

	In this section, the heat equation is discussed as a special case of diffusion equations where the solution of this PDE gives rise to the term $e^{-Ls^{0.5}}$ in the Laplace domain. The generalized form of this term $e^{-(Ls)^{\alpha}}$, the Fractional Diffusive Delay (FDD),  is discussed from a theoretical point-of-view. First, the straightforward frequency domain behavior of the FDD is described. Then, some remarks regarding time domain and some interesting properties are proposed. An inverse transformation of FDD from Laplace domain to time domain and its numerical implementation are not trivial. Therefore, the final part addresses an algorithm to obtain the impulse response of the FDD using numerical tools. 

\subsection{Heat diffusion equation: a source of inspiration}
	A fractional exponential term has already appeared in differential equation theory. The heat diffusion equation is a well-known PDE that represents the heat distribution in a material in time and location:
	\begin{equation}
		u_{t}=\kappa\nabla^{2}u
	\end{equation}
	with $\kappa$ the thermal diffusivity and $u_{t}$ the time derivative of $u(x,t)$. In the case of a semi-infinite, uniform and 1-dimensional rod, $u(x,t)$ is the temperature profile along the $x$-axis and at time $t$, which leads to:
	\begin{equation}\label{eq:1DheatPDE}
	\frac{\partial u(x,t)}{\partial t} = \kappa \frac{\partial^{2}u(x,t)}{\partial x^{2}}
	\end{equation}
	This PDE is accompanied with boundary conditions. With these, the transfer function of this distributed parameter system can be obtained as shown in \cite{Curtain2009}.

	In \cite{Sierociuk2015}, Sierociuk et al. make use of a large number of connected resistors-capacitors (RC) networks as an electro-analogon of this thermal diffusion process. A semi-infinite, thermally conductive rod is modelled as an infinite amount of interconnected RC networks by subdividing the rod in infinitesimally small sections and modeling each of section as a RC network. 
	
	\begin{figure}[t]
		\begin{center}
			\begin{circuitikz}
				\draw (0,0)
				to[R=$R$,i>^=$i(x{,}t)$] (3,0)
				to [C, l_=$C$] (3,-2.5);
				\draw (3,-2.5)
				to[short] (0,-2.5)
				(0,-2.5) to [european voltages,open,v^=$u(x{,}t)$]    (0,0);
				\draw (3.2, -2.5) to [european voltages,open,v=$u(x+dx{,}t) $] (3.2,0);
				\draw (3,0) to (4,0);
				\draw (3.8,0) node[label={above:$i(x+dx{,}t)$}] {};
				\draw (4,0)
				to[R=$R$,i>^=$ $] (7,0) 
				to[C, l_=$C$] (7,-2.5)
				to (3, -2.5);
				\draw[dashed] (7,0) to (8,0);
				\draw[dashed] (7,-2.5) to (8,-2.5);
			\end{circuitikz}
			\caption{Electro-analog model for heat diffusion in an infinitesimally small section of a thermally conductive rod.}
			\label{fig:transmission-line}
		\end{center}
	\end{figure}
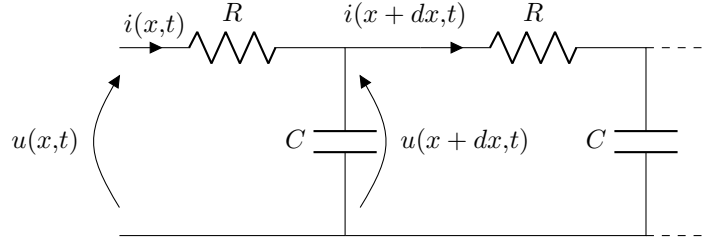

	In Figure \ref{fig:transmission-line}, an infinitesimally small section is depicted as an RC circuit. In this rod section, the continuous material is lumped into two identical subsystems, as shown in Figure \ref{fig:transmission-line}. Using Ohm's law:
	\begin{numcases}{}
		u(x,t) - u(x+dx,t) = Ri(x,t)\\
		i(x,t) - i(x+dx,t) = C\frac{\partial u(x,t)}{\partial t}
	\end{numcases}
	Now, for $dx \rightarrow 0$ this can be rewritten
	\begin{numcases}{}
		\frac{\partial u(x,t)}{\partial x} = Ri(x,t)\label{eq:OhmR}\\
		\frac{\partial i(x,t)}{\partial x} = C\frac{\partial u(x,t)}{\partial t}\label{eq:OhmC}
	\end{numcases}
	which by combining \eqref{eq:OhmR} and \eqref{eq:OhmC} yields an equivalent expression as \eqref{eq:1DheatPDE} with $\kappa = (RC)^{-1}$:
	\begin{numcases}{}
	\frac{\partial^{2}u(x,t)}{\partial x^2} = RC\frac{\partial u(x,t)}{\partial t} & \\
	u(x,0) = 0 & \nonumber\\
	u(0,t) = f(t) & \nonumber
	\end{numcases}
	Given the boundary conditions, the solution of $U(x,s)$ in the Laplace domain is:
	\begin{equation}
		U(x,s) = F(s)\exp\left(-\sqrt{RCs}\cdot x\right)
	\end{equation}
	where $F(s)=\Laplace[f(t)](s)$. So, at an arbitrarily chosen position $x$ the temperature profile in frequency domain is related to the term $\exp(-(Ls)^{\alpha})$ with $\alpha = 0.5$. Notice that this solution, with $\alpha=0.5$, applies to an infinite number of interconnected RC networks.

	Lately, there is a trend to generalize physical laws and their respective PDE's using non-integer order differentials. In the past, these generalizations have been proposed to describe the propagation of plane electromagnetic waves in isotropic and homogeneous, lossy dielectrics, the Maxwell equations, and even Newton's second law \cite{Podlubny1999}. In this line of thought, the authors propose a generalized, fractional order heat diffusion PDE:
	\begin{numcases}{}
		\frac{\partial^{2} u(x,t)}{\partial x^{2}} = L^{2\alpha}\frac{\partial^{2\alpha}u(x,t)}{\partial t^{2\alpha}}, \quad \alpha \in ]0,1[ & \label{eq:generalPDE}\\
		u(x,0) = 0 & \label{eq:generalBC1}\\
		\lim_{s\to \infty}{U(x,s)} = 0 & \label{eq:generalBC2}\\
		U(0,s) = F(s) & \label{eq:generalBC3}
	\end{numcases}
	with the definition of $]a, b[\ = \{x\in\mathbb{R}|a < x < b; a, b\in\mathbb{R}\}$ and with the fractional derivative of order $\mu > 0$ in the Caputo sense:
	\begin{equation}\label{eq:CaputoDerivative}
		\frac{d^{\mu}}{dt^{\mu}} f(t) =
		\begin{cases}
			\frac{1}{\Gamma(m - \mu)}\int_{0}^{t}\frac{f^{(m)}(\tau)d\tau}{(t - \tau)^{\mu + 1 - m}} & m - 1 < \mu < m\\
			\frac{d^{m}}{dt^{m}}f(t) & \mu = m
		\end{cases}
	\end{equation}
	with $\Gamma(z) = \int_{0}^{\infty}t^{z-1}e^{-t}dt$, the gamma-function.
	
	Boundary condition \eqref{eq:generalBC3} is the Laplace transform of the input signal (i.e. temperature or voltage for thermal or electric PDE's respectively) in time at location $x=0$. A step and impulse input signal is given by $F(s) = \frac{1}{s}$ and $F(s) = 1$ respectively. Remark that the PDE in \eqref{eq:generalPDE} resembles the fractional diffusion-wave equation in \cite{Gorenflo1999} ($\mathcal{D} = 1/L^{2\alpha}$, and $\alpha$ is scaled with a scalar) and the time-fractional diffusion equation in \cite{Mainardi2010} ($K_{\beta} = 1/L^{2\alpha}$, and $\beta$ is equal to $2\alpha$). 
	
	The transfer function of the proposed PDE can be found by taking the Laplace transform of \eqref{eq:generalPDE}:
	\begin{equation}
		s^{2\alpha}U(x,s) - u(x,0) = \frac{1}{L^{2\alpha}}\frac{\partial^{2}U(x,s)}{\partial x^{2}}
	\end{equation}
	From boundary condition \eqref{eq:generalBC1} it follows that
	\begin{equation}
		\frac{d^{2}U(x,s)}{{dx}^{2}} - (Ls)^{2\alpha}U(x,s) = 0
	\end{equation}
	This is an ordinary differential equation and the solution has the form:

	\begin{equation}
		U(x,s) = C_{1}(s)\exp\left((Ls)^{\alpha}x\right) + C_{2}(s)\exp\left(-(Ls)^{\alpha}x\right)
	\end{equation}
	The boundary conditions \eqref{eq:generalBC2} and \eqref{eq:generalBC3} give rise to the following solution
	\begin{equation}\label{eq:generalizedSolution}
		U(x,s) = F(s)\exp\left(-(Ls)^{\alpha}x\right)
	\end{equation}

	This generalization of the heat diffusion equation in \eqref{eq:generalPDE}, with the well-known expression for $\alpha=0.5$, generates a solution \eqref{eq:generalizedSolution} that contains the term $\exp{\left(-(Ls)^{\alpha}\right)}$. This concludes the premise that it is possible to construct the fractional exponential term with an adapted heat diffusion PDE, which is extended with fractional calculus.

\subsection{FDD: frequency domain}
As mentioned in \cite{Juchem2019}, the FDD can be rewritten to obtain the frequency response as:
\begin{equation}
\exp{\left(-(Ls)^{\alpha}\right)}\lvert_{s=j\omega} =  \exp\left(-(L\omega)^{\alpha}\cos\left(\frac{\alpha\pi}{2}\right)\right) \cdot \left[\cos\left((L\omega)^{\alpha}\sin\left(\frac{\alpha\pi}{2}\right)\right) - j\sin\left((L\omega)^{\alpha}\sin\left(\frac{\alpha\pi}{2}\right)\right)\right]
\end{equation}
where the modulus and phase are:
\begin{align}
	M &= \exp\left(-(L\omega)^{\alpha}\cos\left(\frac{\alpha\pi}{2}\right)\right)\label{eq:modulus}\\
	\phi &= -(L\omega)^{\alpha}\sin\left(\frac{\alpha\pi}{2}\right)\label{eq:phase}
\end{align}

As shown in equations \eqref{eq:modulus} and \eqref{eq:phase}, varying $L$ influences both modulus and phase. This coupling between modulus and phase is an interesting feature of the FDD. For completeness sake, there is no coupling between the modulus and phase for any parameter for $\alpha=1$, i.e. integer order delay, as the modulus and phase simplify to $M_{\alpha=1} = 1$ and $\phi_{\alpha=1} = -L\omega$. The integer order exponential is a special case in which the dead time $L$ does not appear in the equation of the modulus.

\subsection{FDD: time domain - theoretical derivation of the impulse response}
The FDD system's transfer function is relatively easy to analyze in frequency domain with graphical representations like Bode plots, Nyquist plots, etc. However, in time domain the analysis is not straightforward. No direct inverse Laplace transform $\Laplace^{-1}[G(s)](t)$ for the transfer function exists and the technique of partial fraction decomposition is not an option. To find the impulse response of the FDD, its transfer function is expanded according to Taylor series:
\begin{align}
	\Laplace^{-1}\left[\exp\left(-(Ls)^{\alpha}\right)\right](t) &= \Laplace^{-1}\left[\sum_{i=0}^{\infty}\frac{(-(Ls)^{\alpha})^i}{i!}\right](t)\nonumber\\
	&= \sum_{i=0}^{\infty} \frac{(-L^{\alpha})^{i}}{i!}\Laplace^{-1}\left[s^{\alpha i}\right](t)\nonumber\\
	& = \sum_{i=0}^{\infty} \frac{(-L^{\alpha})^{i}}{i!}\frac{t^{-\alpha i-1}}{\Gamma(-\alpha i)}\label{eq:TaylorExp}
\end{align}
This last step is obtained using the inverse Laplace transforms in \cite{Monje2010} where the Caputo definition of a fractional derivative is used (see \eqref{eq:CaputoDerivative}). The inverse Laplace of a fractional integrator for $\alpha \in \mathbb{R}$ is given by:
\begin{equation}
\Laplace^{-1}\left[\frac{1}{s^{\alpha}}\right](t) = \frac{t^{\alpha-1}}{\Gamma(\alpha)}
\end{equation}
Notice that the sum in \eqref{eq:TaylorExp} converges for all $\alpha > 0$. For the special case of $\alpha=0.5$ a closed-form expression of the sum can be found:
\begin{equation}\label{eq:impulseSqrt}
\Laplace^{-1}\left[\exp\left(-L\sqrt{s}\right)\right](t) = \frac{L}{2\sqrt{\pi}t^{3/2}}\exp\left(\frac{-L^2}{4t}\right)
\end{equation}
For other $\alpha$ the authors were not able to find a closed-form expression.

The range of the fractional exponent $\alpha$ has been under discussion in literature varying between $]0,1[$ and $]0,2[$ \cite{Monje2010}. The ongoing discussion handles the range of the fractional component for traditional poles and zeros in transfer functions. In the definition of the FDD, the fractional exponent $\alpha$ is found in the argument of the exponential function. In Lemma \ref{prop:boundsAlpha}, it is shown that $\alpha$ is restricted to $]0,1[$ for the FDD.

\begin{lem}\label{prop:boundsAlpha}
	Given $\alpha \in \mathbb{R}$, $L\in \mathbb{R}_{0}^{+}$ and $t > 0$, then the impulse response of $\exp\left(-(Ls)^{\alpha}\right)$ exists if and only if $\alpha \in ]0,1[$.
\end{lem}
\begin{proof} 
The proof is delivered in two steps:
\begin{enumerate}
\item For $\alpha < 0$, the Taylor expansion from \eqref{eq:TaylorExp} does not converge (for $t \geq 1$).\\
\item For $\alpha > 0$, the sum converges for all $t$. Given the definition of the Wright function \cite{Gorenflo1999}:
	\begin{equation}
	W_{\lambda, \mu}(z) := \sum_{n=0}^{\infty}\frac{z^{n}}{n!\Gamma(\lambda n+\mu)}, \quad \lambda > -1, \mu \in \mathbb{C}
	\end{equation}
which is an entire function for $z$. Rewriting \eqref{eq:TaylorExp} with the Wright function gives:
	\begin{equation}\label{eq:sumAsWright}
	\sum_{i=0}^{\infty}\frac{-L^i}{i!}\frac{t^{-\alpha i-1}}{\Gamma(-\alpha i)} = -\frac{W_{-\alpha, 0(Lt^{-\alpha})}}{t}
	\end{equation}
	In \eqref{eq:sumAsWright}, $\lambda=-\alpha$. From the condition $\lambda > -1$ follows that $\alpha < 1$. 
\end{enumerate}
\end{proof} 

\subsection{FDD: time domain - numerical implementation }\label{ssec:numSim}
In the previous section, it is shown that the impulse response of the FDD can be expressed as an infinite sum \eqref{eq:TaylorExp}. As stated before, only for $\alpha=0.5$ a closed-form expression can be found for this infinite summation \eqref{eq:impulseSqrt}. For all other values of $\alpha$, the infinite summation has to be approximated by terminating at a finite index such that the remainder is minimal. Some problems arise from a numerical implementation, which will be addressed here. The algorithm to create the numerical solution of the infinite sum is created in MATLAB$^{\tiny\textregistered}$, but could easily be translated to any other programming language.

\subsubsection{The problem of convergence and its solution} Expression~\eqref{eq:TaylorExp} contains the Gamma function $\Gamma(z)$ which is known to have singularities, so small numerical errors can lead to unexpected behavior. Therefore, Lemma~\ref{prop:FDD} is defined to end the infinite sum in a proper manner and minimize the remainder. Before Lemma~\ref{prop:FDD} is explained, two additional theorems are needed for this lemma.
\begin{theo}[Final value theorem \cite{Oppenheim1997}]\label{Theorem:FinalValue}
	If $g(t)$ is bounded on $]0,\infty[$ and $\lim\limits_{t\to\infty} g(t)=\rho$ with $\rho < \infty$, then 
	\begin{equation}
		\lim\limits_{t\to\infty}g(t)=\lim\limits_{s\to 0}sG(s)\nonumber
	\end{equation}
	with $G(s) = \Laplace[g(t)](s)$.
\end{theo}

\begin{theo}[Cauchy criterion \cite{Abbott2001}]\label{Theorem:Cauchy}
	A series $\sum_{i=1}^{\infty}a_{i}$ converges
	\begin{equation}
		\Longleftrightarrow \forall \epsilon > 0: \exists N \in \mathbb{N} \Rightarrow |a_{n+1} + a_{n+2} + \dots + a_{n+p}| < \epsilon,\qquad n > N, p \geq 1 \nonumber
	\end{equation}
\end{theo}

\begin{lem}\label{prop:FDD}
	Given $t>0$, $\alpha \in \mathbb{R}_{]0, 1[}$, $L\in \mathbb{R}_{0}^{+}$ and $f(t)$ is the impulse response of $F(s) = \exp\left(-(Ls)^{\alpha}\right)$, then $A(\alpha)\equiv 1$, with $A=\int_{0}^{\infty}f(t)dt$.
\end{lem}
\begin{proof} 
The surface underneath the impulse response is given by:
\begin{align}
	A(\alpha) &= \int_{0}^{\infty}\Laplace^{-1}[\exp\left(-(Ls)^{\alpha}\right)](t)dt\nonumber\\
	& = \lim\limits_{t\to\infty}\Laplace^{-1}[Q(s, \alpha)](t)\nonumber
\end{align}
with $Q(s, \alpha) = \frac{1}{s}\cdot\exp\left(-(Ls)^{\alpha}\right)$. 

Due to the final value theorem (Theorem \ref{Theorem:FinalValue})
\begin{equation}\label{eq:FinalValueTheorem}
	A(\alpha) = \lim\limits_{s\to 0} sQ(s,\alpha)
\end{equation}
if $\Laplace^{-1}[Q(s, \alpha)]$ is bounded and has a finite limit, or if $f(t)$ is bounded and has a finite limit.
If $t>0$, $\alpha \in \mathbb{R}_{]0, 1[}$ and $L\in \mathbb{R}_{0}^{+}$, then the sum $f(t)$ converges according Lemma \ref{prop:boundsAlpha}. Due to the Cauchy convergence criterion (Theorem \ref{Theorem:Cauchy}) $f(t)$ is bounded $\forall t > 0$.
From \eqref{eq:TaylorExp}, $\lim\limits_{t\to \infty}f(t) = 0$ if $\alpha \geq \frac{-1}{i}$ with $i\in ]0,\infty[$, which is always true, because $\alpha \in ]0, 1[$.

Therefore, due to \eqref{eq:FinalValueTheorem}:
\begin{equation}
	A(\alpha) = 1 \nonumber
\end{equation}
which is independent of $\alpha$.
\end{proof} 

Lemma~\ref{prop:FDD} will be used to stop the infinite summation by finding the numerical threshold for which the remainder is minimized. It is paramount to find the amount of terms $N$ for which the end result has been approximated with a certain tolerance. If the sum of the $N$ first terms is close to the theoretical surface of 1, terms with higher index are not superfluous.
\begin{lem}\label{prop:energy}
	Given $t>0$, $\alpha \in \mathbb{R}_{]0, 1[}$, $L\in \mathbb{R}_{0}^{+}$ and $f(t)$ is the impulse response of $F(s) = \exp\left(-Ls^{\alpha}\right)$, then the energy of the function depends on $\alpha$.
\end{lem}
\begin{proof} 
	The energy of a function $f(t)$ is defined by:
	\begin{equation}
		E := \int_{-\infty}^{\infty}f^{2}(t)dt
	\end{equation}
	Also here, it is preferred to work in frequency domain as a closed-form equation exists. For this, Parceval's theorem is used, which leads to a frequency domain representation of the energy:
	\begin{equation}
		E = \frac{1}{2\pi}\int_{-\infty}^{\infty}|F(j\omega)|^2d\omega
	\end{equation}
	if $f(t)$ has finite energy. According to Lemma \ref{prop:boundsAlpha}, this is the case for $\alpha \in ]0, 1[$.
	This means that from \eqref{eq:modulus}
	\begin{equation}
	E_{\alpha} = \frac{1}{2\pi}\int_{-\infty}^{\infty}\exp\left(-2L\omega^{\alpha}\cos\left(\frac{\alpha\pi}{2}\right)\right)d\omega
	\end{equation}
	The question is whether a difference in energy is observed for different $\alpha_1, \alpha_2 \in ]0, 1[$, $\alpha_1\neq\alpha_2$:
	\begin{align}
	\Delta E &= E_{\alpha_2} - E_{\alpha_1}\nonumber\\
	&=\exp(-2L)\int_{-\infty}^{\infty}\exp\left[\omega^{\alpha_2}\cos\left(\frac{\alpha_2\pi}{2}\right)\right] - \exp\left[\omega^{\alpha_1}\cos\left(\frac{\alpha_1\pi}{2}\right)\right]d\omega \stackrel{?}{=} 0
	\end{align}
	If the integrand $g(\omega)$ is piecewise continuous and uneven ($g(-\omega) = -g(\omega)$), the difference in energy identifies with zero. However, $g(-\omega) \neq -g(\omega)$, which proves that the energy is not equal for different $\alpha$.
\end{proof}

In Lemma~\ref{prop:energy} a proof is given that $\alpha$ will affect the energy of the signal. This fact, combined with the result of Lemma~\ref{prop:FDD}, gives insight in the expected impulse responses before simulating.

Equation \eqref{eq:TaylorExp} is evaluated in the time interval $t=]0, T_{max}]$. To understand the consequence of cutting-off the summation, it is important to understand the effect of each sub-term on the error. Therefore, each sub-term of the summation \eqref{eq:TaylorExp} is analyzed.

\begin{figure*}[t]
	\centering
	\begin{subfigure}[b]{0.325\textwidth}
		\centering
		\includegraphics[width=\textwidth]{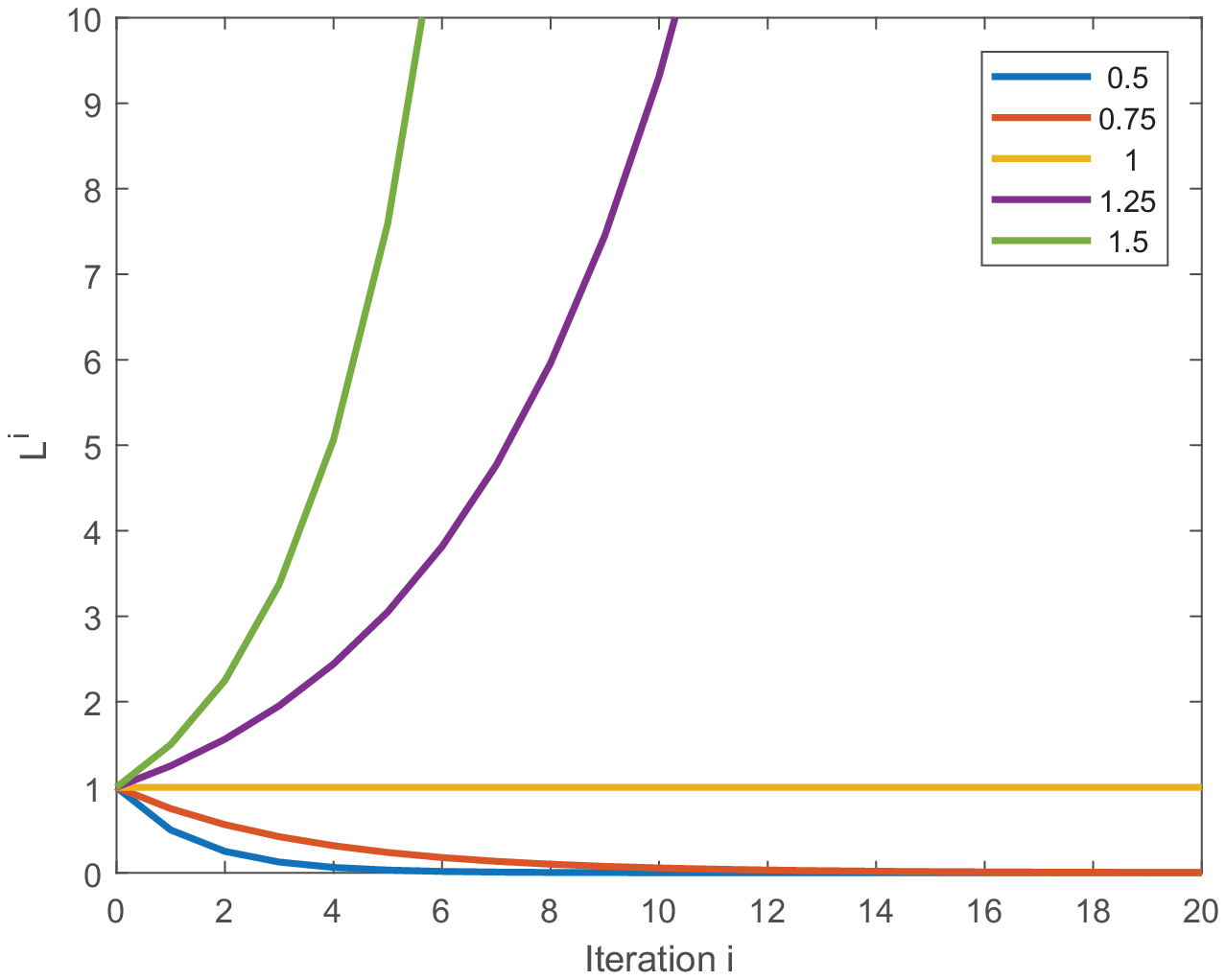}
		\caption[]%
		{}    
		\label{fig:term1}
	\end{subfigure}
	\hfill
	\begin{subfigure}[b]{0.325\textwidth}  
		\centering 
		\includegraphics[width=\textwidth]{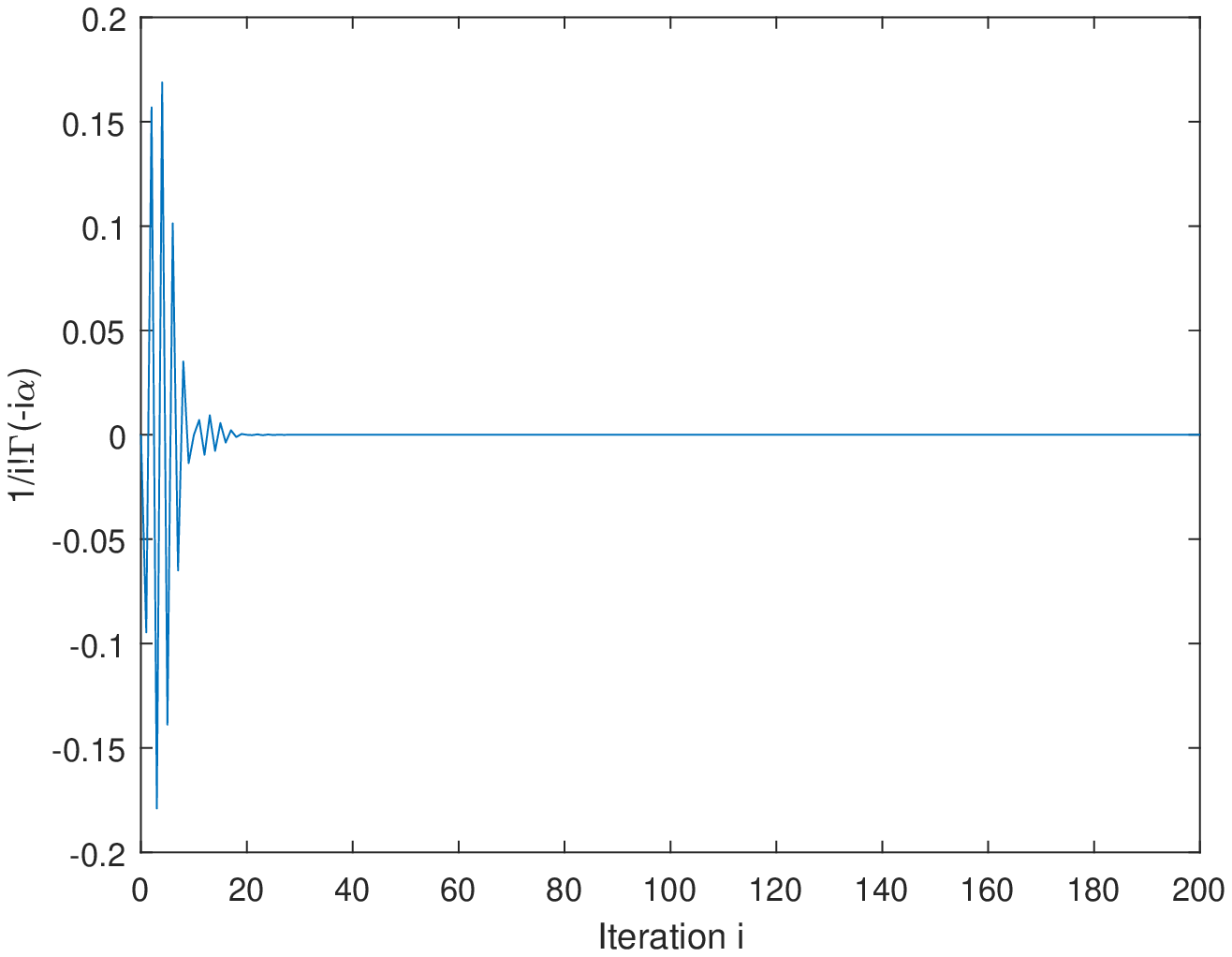}
		\caption[]%
		{}    
		\label{fig:term2}
	\end{subfigure}
	\hfill
	\begin{subfigure}[b]{0.325\textwidth}   
		\centering 
		\includegraphics[width=\textwidth]{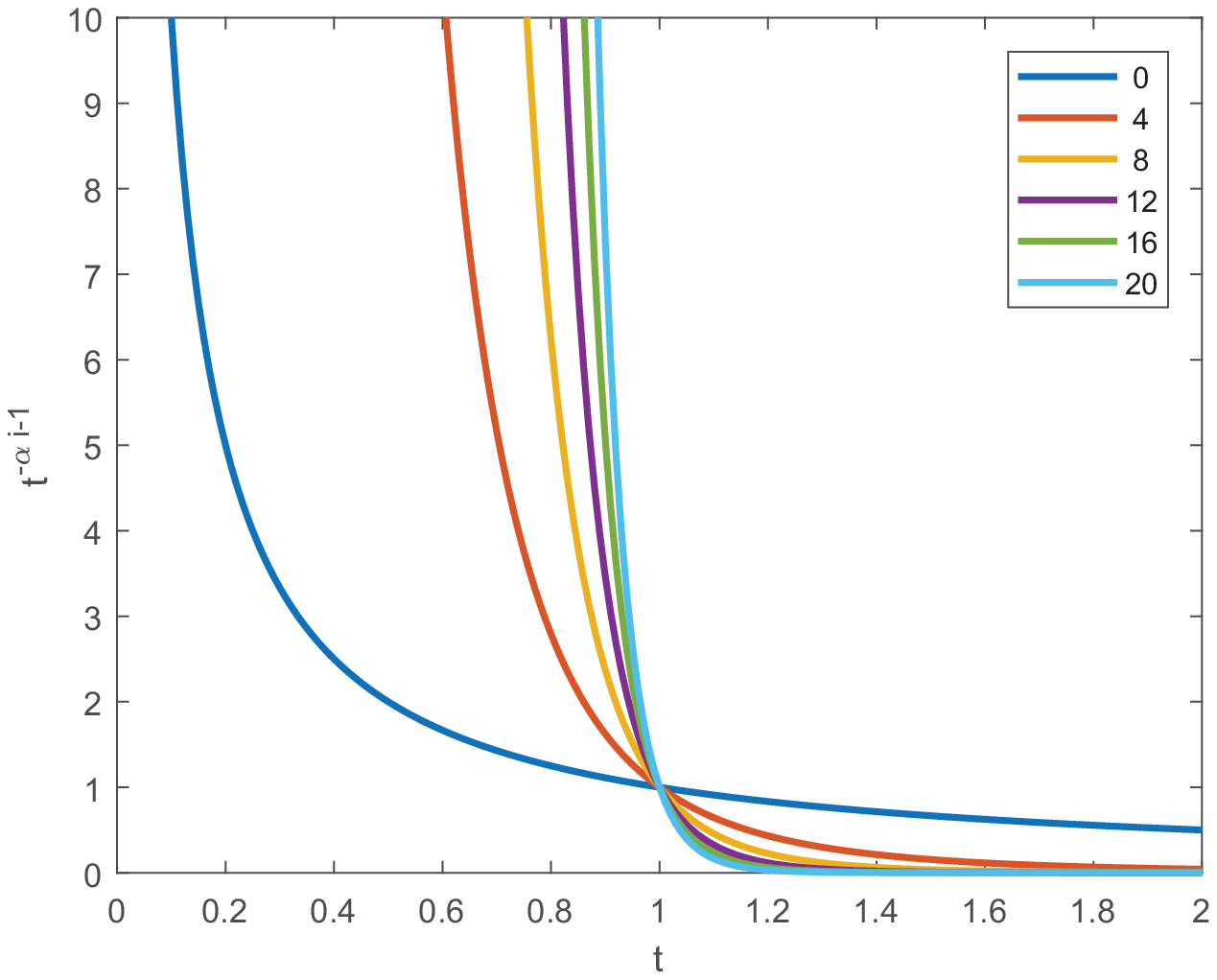}
		\caption[]%
		{}    
		\label{fig:term3}
	\end{subfigure}
	\caption{Each term of the sum of \eqref{eq:TaylorExp} consists of sub-terms. To better understand the influence of each parameter on the entire term its behavior is plotted for different values.} 
	\label{fig:summation}
\end{figure*}

\begin{itemize}
\item First the term $L^i$ is examined. For $L<1$ the term gradually decreases for increasing $i$, for $L=1$, the term is a constant 1, and for $L>1$ it gradually increases for increasing $i$ (see Figure \ref{fig:term1}). 

\item Next, the denominator contains $i!$. This term rises for increasing $i$. In MATLAB$^{\tiny\textregistered}$ this term is quickly hitting the limitations of the \textit{Double-Precision Floating Point} data type (i.e. from $170!$ onwards). MATLAB$^{\tiny\textregistered}$ converts it to an \textit{Inf}, which is detrimental to the end result.

\item The third term is an array, namely $t^{-\alpha i - 1}$. The exponent $-\alpha i - 1 < 0 $ for all $i$, which means that for  $t > 1$ the term will go to zero for increasing $i$. For $t=1$ the terms is a constant, i.e. one. For $t < 1$, the term explodes, creating numerical issues (see Figure \ref{fig:term3}).
 
\item The last term is the gamma-function $\Gamma(-\alpha i)$. MATLAB$^{\tiny\textregistered}$ features a built-in function \textit{gamma()} which enables the evaluation of the gamma-function in a certain point. The argument of the gamma-function, $-\alpha i \leq 0$ as $\alpha \in ]0, 1[$. For increasing $i$ the term will converge to zero (see Figure \ref{fig:term2}). However, if $\alpha i \rightarrow n$ with $n \in \mathbb{N}$, then $|\Gamma(-\alpha i)| \rightarrow \infty$. For increasing $i$, this singularity has a diminishing effect on its direct neighbors. This means that for $\alpha i\in\mathbb{N}$, the entire term can be assumed to equal zero. In other cases, the term is suffering from limitations regarding memory as well. Therefore, it is reasonable to use a new data type which has a larger amount of precision, namely the Variable-Precision Arithmetic (VPA) function.
\end{itemize}
In MATLAB$^{\tiny\textregistered}$, the Variable-Precision Arithmetic (VPA) function allows to increase the number of significant digits using a symbolic notation. However, the speed to execute mathematical operations decreases drastically. Consequently, this data type should only be used when necessary. The two terms in the denominator have an opposing feature when $i$ increases, as one increases in size and the other decreases. Therefore, both terms can be calculated after being transformed to VPA. This leads to a smaller number, which means both can be transformed back to \textit{double} before executing the remainder of the calculations, which makes the simulation much shorter. The entire algorithm to find the impulse response $\exp\left(-(Ls)^{\alpha}\right)$ is given in Algorithm \ref{alg:FDD}. The function requires six inputs. The first five are described before. Parameter $P$ is a parameter that improves the quality of the impulse response by discerning relevant peaks versus singularities, caused by numerical inaccuracies. The better the estimate of the peak value (which must be certainly higher than the actual peak value) , the better the quality of the impulse response. The algorithm is optimized to have the shortest calculation time. Default input parameters are $\{N, P\} = \{200, 10\}$. For increasing $L$, both parameters need to be lowered according to Lemma~\ref{prop:FDD}.

\begin{breakablealgorithm}
	\caption{Fractional Diffusive Delay impulse response simulation algorithm}
	\label{alg:FDD}
	\begin{algorithmic}[1]
		\Ensure $F(t)$ - approximation of the impulse response of FDD
		\Function{FDD}{$L$, $\alpha$, $T_{s}$, $T_{max}$, $N$, $P$}
			\State $t \leftarrow T_{s}:T_{s}:T_{max}$
			\Comment{Generate discrete time vector.}
			\For {$i \leftarrow 1, N$}
				\If {$i\alpha \in \mathbb{N}$}
				\Comment{}
					\ForAll {$k \in T(k)$}
						\State $T_{i}(k) \leftarrow 0$
						\Comment{If $i\alpha \in \mathbb{N},\ \Gamma(-i\alpha) \rightarrow \infty \Rightarrow T_{i}(k) \leftarrow 0, \forall k$.}
					\EndFor
				\Else
					\Comment{Calculate the $i^{th}$ term by evaluating the different sub-terms.}
					\State $f_{i} \leftarrow i!$
					\State $g_{ia} \leftarrow \Gamma(-i\alpha)$
					\State $L_{ia} \leftarrow (-1)^{i}L^{\alpha i}$
					\State $t_{pow}(k) \leftarrow t(k)^{-\alpha i - 1}$
					\State $T_{i}(k) \leftarrow \frac{L_{ai}t_{pow}(k)}{f_{i}g_{ia}}$
					\Statex
					\If {($f_{i} ==$ Inf) $\vee\ (g_{ia} ==$ Inf) $\vee\ (L_{ia} ==$ Inf)}
					\Comment{Increase precision of the sub-terms.}
						\State $f_{i,vpa} \leftarrow \vpa(i)!$
						\State $g_{ia,vpa} \leftarrow \Gamma(-\vpa(i)\alpha)$
						\State $L_{ia,vpa} \leftarrow (-1)^{i}\vpa(L)^{\alpha i}$
						\State $C_{vpa} \leftarrow \frac{L_{ia,vpa}}{f_{i,vpa}g_{ia,vpa}}$
						\State $T_{i}(k) \leftarrow \double(C_{vpa}t_{pow}(k))$
						\Comment{The subterms are multiplied to decrease numerical errors. Afterwards, the term is converted back to \textit{double} precision.}
						\Statex
						\State $q \leftarrow \find(t_{pow} ==$ Inf$)$
						\Comment{a vector with the indexes fulfilling the statement.}
						\ForAll {$j \in q$}
							\State $T_{i}(j) \leftarrow \double(C_{vpa}\vpa(t(j))^{-\alpha i - 1})$
						\EndFor
					\EndIf
 				\EndIf
 				\State $F(k) \leftarrow F(k) + T_{i}(k)$
 				\Comment{Add $i^{th}$ term to the sum.}
 			\EndFor
			\Statex
 			\State $q_{1} \leftarrow \find(\abs(F(k)) > P)$
 			\Comment{Removing numerical errors due to cutting off the infinite sum.}
 			\If {$q_{1} \nsubseteq \emptyset$}
	 			\For {$i \rightarrow 1, q_{1}(end)$}
	 				\State $F(i) \leftarrow 0$
	 			\EndFor
	 		\EndIf
 			\State $q_{2} \leftarrow \find(F(k) < 0)$
 			\If {$q_{2} \nsubseteq \emptyset$}
 				\For {$i \rightarrow 1, q_{2}(end)$}
 					\State $F(i) \leftarrow 0$
 				\EndFor
 			\EndIf
		\EndFunction	
	\end{algorithmic}
\end{breakablealgorithm}

\subsubsection{Numerical simulations}
 In Figure \ref{fig:impulseResp}, the impulse responses are plotted for varying $\alpha$ (Figure \ref{fig:impulseRespAlpha}) and varying L (Figure \ref{fig:impulseRespL}). Notice that
\begin{equation}\label{eq:dirac}
\lim\limits_{\alpha\to 1}\Laplace^{-1}[\exp\left(-(Ls)^{\alpha}\right)](t) = \delta(t-L^{\alpha})
\end{equation}
From Figure \ref{fig:impulseResp} the choice of the FDD's name becomes clear. The fractional order exponential term leads to a mix of a delay and diffusion effect. In Figure \ref{fig:impulseRespAlpha}, the diffusive effect increases for $\alpha$ going from $0$ to $0.5$ and starts decreasing again for increasing $\alpha$ until no diffusive effect is observed for $\alpha \rightarrow 1$, which leads to the dirac response, as in \eqref{eq:dirac}. In Figure \ref{fig:impulseRespL}, increasing $L$ results in a longer delay of the response to an impulse. This can be explicitly shown with the time-scaling property of the Laplace transform:
\begin{cor}[Laplace transform's time-scaling property]\label{cor:timeScaling}
	If $x(t) \longleftrightarrow X(s)$, then
	\begin{equation}
		t \mapsto x(at), a\in\mathbb{R}, a\neq 0 \longleftrightarrow s\mapsto \frac{1}{|a|}X\bigg(\frac{s}{a}\bigg)
	\end{equation}
\end{cor}
Take $F(s) = e^{-s^{\alpha}}$ and $G(s) = F(sL) = \exp{(-(sL)^{\alpha})}$. According to Lemma \ref{cor:timeScaling}, $g(t) = \frac{1}{L}f(\frac{t}{L})$ with $g(t) \longleftrightarrow G(s)$ and $f(t) \longleftrightarrow F(s)$. This shows that the parameter $L$ is equivalent with a time-scaling and an amplitude-scaling, which is independent of $\alpha$. This again confirms the results in Figure \ref{fig:impulseRespL}.

\begin{figure}[h!]
	\centering
	\begin{subfigure}{0.47\textwidth}
		\centering
		\includegraphics[width=\linewidth]{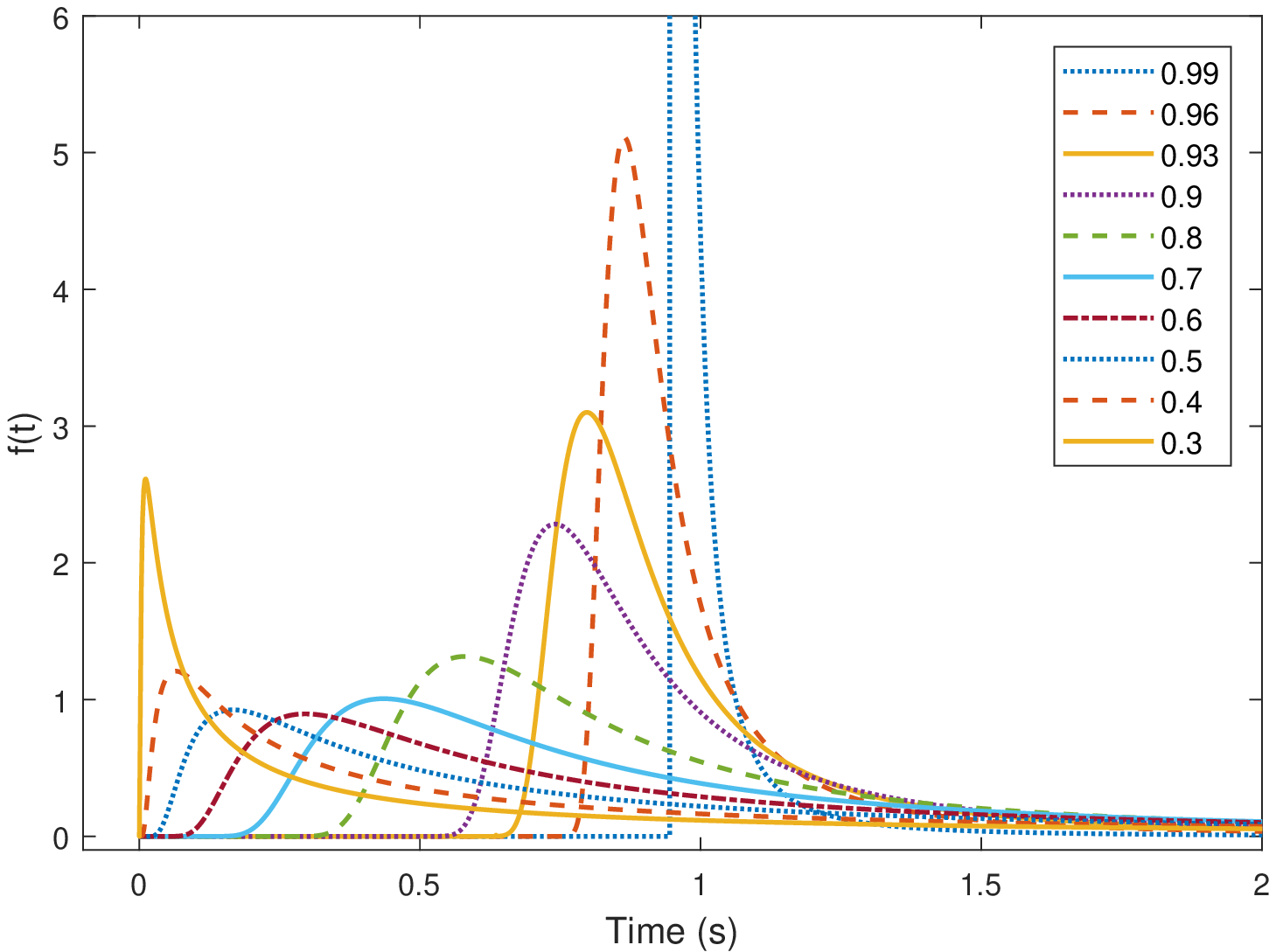}
		\caption{}
		\label{fig:impulseRespAlpha}
	\end{subfigure}
	\begin{subfigure}{0.47\textwidth}
		\centering
		\includegraphics[width=\linewidth]{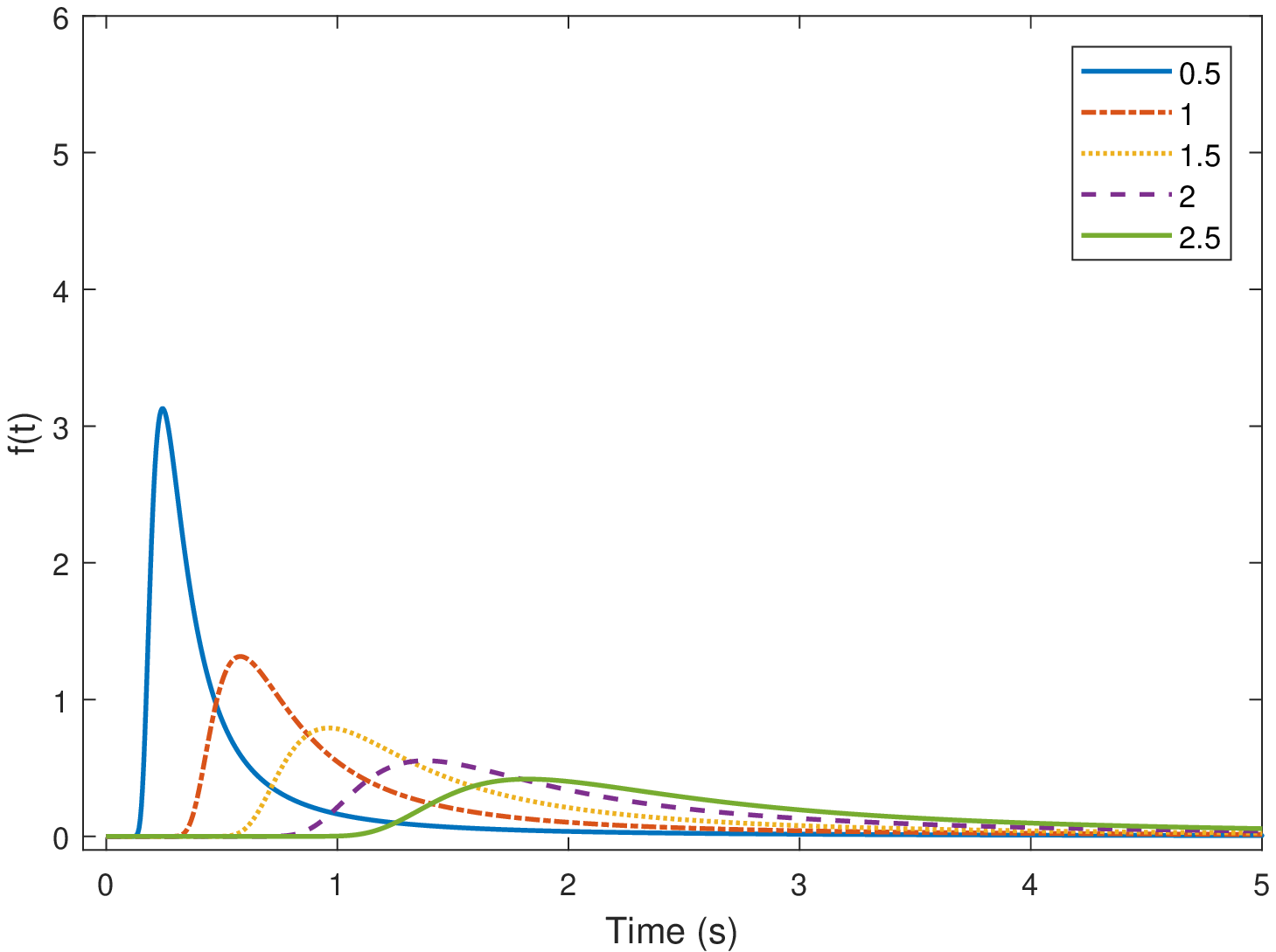}
		\caption{}
		\label{fig:impulseRespL}
	\end{subfigure}
	\caption{The impulse response for FDD for (a) varying $\alpha$ with $L=1$ and (b) varying delay $L$ with $\alpha=0.8$.}
	\label{fig:impulseResp}
\end{figure}

\section{Methodology}\label{sec:M&M}
\setcounter{section}{3}
\setcounter{equation}{0}

In this section, the authors propose a new model based on the previous observations in this paper. This new model combines the idea of a First Order Plus Dead Time approximation with the FDD resulting in a First Order Plus Fractional Diffusive Delay (FOPFDD) model, capable of identifying delay-dominant, higher-order processes.

To find the FOPFDD model of a system, a methodology is provided to find the parameters as proposed in \cite{Juchem2019}. Also, a methodology to obtain a time domain analysis is presented in this paper. The FOPFDD model is expressed as:

\begin{equation}\label{eq:FOPFDD}
F(s)=\frac{K}{\tau s+1}e^{-(Ls)^\alpha}.
\end{equation}
Notice that the choice of introducing $L^{\alpha}$ leads to a more physically intuitive interpretation of parameter $L$ as it can be expressed in units of seconds now.

\subsection{Model fitting: an optimization approach}

The model fitting method is formulated as a multi-objective optimization problem. The rationale is to find a model that minimizes the error between frequency response of the original and simplified model.

It is assumed that the frequency response function of the system that needs to be modeled, is available, either through identification techniques or mathematical modeling. The model is fitted on the system's frequency response by minimizing the error between the magnitude and phase. A Pareto front \cite{Miettinen1999} is obtained from which the trade-off between minimizing the magnitude or the phase error can be observed (see Figure \ref{fig:n6Pareto}). Then, the normalized error for both objectives is minimized, such that the optimal fit is found. This methodology is earlier described in \cite{Juchem2019}. The minimization problem leads to the model parameters $\{K, \tau, L, \alpha\}$. 

\subsection{Time domain solution}
A good fit in frequency domain is important, but many processes are evaluated in time domain. For instance, a step response gives insight into the dynamical behavior of the system in time. Therefore, the model in \eqref{eq:FOPFDD} is subdivided into two parts:
\begin{align}
	F_1(s) &= \frac{K}{\tau s+1}U(s)\nonumber\\
	F_2(s) &= \exp\left(-(Ls)^{\alpha}\right)\nonumber
\end{align}
Given that $F_1(s) = \Laplace\left[f_1(t)\right](s)$ and $F_2(s) = \Laplace\left[f_2(t)\right](s)$ and remarking that $F(s) = F_1(s)\cdot F_2(s)$, it can be stated that $f(t) = f_1(t)*f_2(t)$ with $*$ the convolution, due to duality. Signal $U(s)$ is the input of the model. In the case of a step response $U(s) = 1/s$. This means that 
\begin{equation}\label{eq:convolution}
	f(t) = \int_{-\infty}^{\infty}f_1(\kappa)\cdot f_2(t-\kappa)d\kappa
\end{equation}
The time response $f_{1}(t)$ can be easily obtained as this is the solution of a first order, integer order differential equation and $f_2(t)$ is found with Algorithm \ref{alg:FDD}. With this algorithm, the impulse response can be obtained for a given $L$ and $\alpha$. With \eqref{eq:convolution}, the step response of the FOPFDD model is obtained.

\section{Results}\label{sec:Results}
\setcounter{section}{4}
\setcounter{equation}{0}

\subsection{Bench mark: RC ladder network}

The series RC network that is mentioned earlier is a perfect example of a delay-dominant, higher-order system, which consists of a large number of subsystems. The system is depicted in Figure \ref{fig:ladderNetwork}. The input of the system is the voltage $V_{in}$ and the output is $V_{n}$.

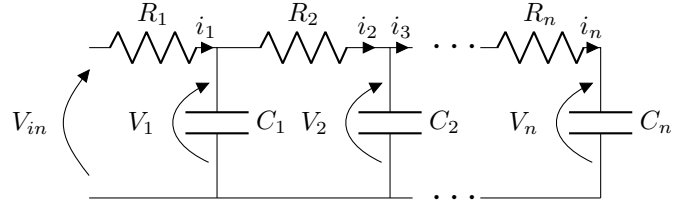
\begin{figure}[h]
	\begin{center}
		\begin{circuitikz}
			\draw (0,0)
			to[R=$R_{1}$,i=$i_1$] (1.7,0)
			to[C=$C_1$,v_<=$V_{1}$] (1.7,-2);
			\draw (1.7,-2)
			to[short] (0,-2)
			(0,-2) to [european voltages,open,v^=$V_{in}$]    (0,0);
			\draw (1.7,0)
			to[R=$R_2$,i=$i_2$] (4,0)
			to[C=$C_2$,v_<=$V_{2}$] (4,-2);
			\draw (4,-2)
			to[short] (1.7,-2);
			\draw (4,0)
			to[short,i=$i_3$] (4.3,0);
			\draw (4,-2)
			to[short,i] (4.3,-2);
   			\draw (4.6,0) --  (5.2,0)node[midway,scale=1.5,fill=white]{$\cdots$};
			\draw (4.6,-2) -- (5.2,-2)node[midway,scale=1.5,fill=white]{$\cdots$};
			\draw (5.2,0)
			to[R=$R_n$,i=$i_n$] (6.8,0)
			to[C=$C_n$,v_<=$V_{n}$] (6.8,-2);
			\draw (6.8,-2)
			to[short] (5.2,-2);
		\end{circuitikz}
		\caption{RC-circuit}
		\label{fig:ladderNetwork}
	\end{center}
\end{figure}

This system can be rewritten in state-space form based on Kirchoff's current and voltage laws:
\begin{equation}
	\begin{aligned}
		\text{KCL}&\rightarrow
		\left\{
		\begin{array}{ll}
			\begin{aligned}
				i_{p}&=i_{p+1}+C_{p}\dot{V}_{p} \quad p \in [1,n-1]\\
				i_{n}&=C_{n}\dot{V}_{n}
			\end{aligned}
		\end{array}
		\right.
		\\\text{KVL}&\rightarrow
		\left\{
		\begin{array}{ll}
			\begin{aligned}
				V_{p-1}&=R_{p}i_{p}+V_{p} \quad p \in [2,n]\\
				V_{in}&=R_{1}i_{1}+V_{1}
			\end{aligned}\\
		\end{array}
		\right.
	\end{aligned}
	\label{eq:KLs}
\end{equation}
	
Combining both yields the dynamics in the $V_{p}$'s:
	
\begin{equation}
	\left\{
	\begin{array}{ll}
		\begin{aligned}
		V_{p-1}&=R_{p}\sum_{j=p}^{n}C_{i}\dot{V}_{i}+V_{p} \quad p \in [2,n]\\
		V_{in}&=R_{1}\sum_{j=1}^{n}C_{i}\dot{V}_{i}+V_{1}
		\end{aligned}\\
	\end{array}
	\right.
	\label{eq:dynamics}
\end{equation}
	
By defining the state $\mathbf{x}=\begin{bmatrix} V_{1} & V_{2} & \dots & V_{n} \end{bmatrix}^{T}\in \mathbb{R}^{n\times1}$, input $u=V_{in}$ and output $y=V_{n}$, \eqref{eq:dynamics}  can be rewritten in the implicit state space model:

\begin{equation}
	\begin{aligned}
		E\dot{\mathbf{x}}&=A\mathbf{x}+Bu \\
		y&=C\mathbf{x} \label{eq3}
	\end{aligned}
\end{equation}
	
with $B=\begin{bmatrix} 1 & 0 & \dots & 0 \end{bmatrix}^{T} \in\mathbb{R}^{n\times 1}$, $C=\begin{bmatrix} 0 & 0 & \dots & 1 \end{bmatrix}\in\mathbb{R}^{1\times n}$ and  $E$ and $A$$ \in \mathbb{R}^{n \times n}$:
\[
	E = \begin{bmatrix} 
	R_{1}C_{1} & R_1C_{2} & R_{1}C_{3}  &\dots & R_{1}C_{n} \\
	0 & R_2C_{2} &R_{2}C_{2} &\dots & R_{2}C_{n} \\
	0 & 0 &R_{3}C_{2} &\dots & R_{3}C_{n} \\
	\vdots & \vdots &\vdots &\ddots &\vdots&\\
	0 &  0 & \dots&0 & R_{n}C_{n} 
	\end{bmatrix}
\]
	
\[
	A = \begin{bmatrix} 
	-1 & 0 & 0 & 0& \dots & 0 \\
	1 & -1 & 0 &0 &\dots & 0 \\
	0 & 1 & -1 & 0 &\dots & 0 \\
	\vdots & \vdots &\vdots &\vdots & \ddots &\vdots&\\
	0 & 0 & 0 &\dots & 1& -1 \\
	\end{bmatrix}
\]
	
$E$ is a non-singular, upper-triangular matrix because of the non-zero product of all diagonal elements. The implicit state space can be rewritten in the standard form:
	
\begin{equation}
	\begin{aligned}
		\dot{\mathbf{x}}&=\overbrace{E^{-1}A}^{\hat{A}}\mathbf{x}+\overbrace{E^{-1}B}^{\hat{B}}u \\
		y&=C\mathbf{x}
	\end{aligned}
\end{equation}

This can be rewritten as a transfer function:

\begin{equation}\label{eq:stateSpaceTF}
	G(s) = C(sI-\hat{A})^{-1}\hat{B}
\end{equation}

Due to the choice of the input ($V_{in}$) and the output ($V_{n}$), a network with $n$ series subnetworks consists of $n$ poles and no zeros. To obtain the complete model, it means that $n$ parameters need to be identified, which is highly inconvenient for large $n$. In the subsequent analysis, three models with a limited amount of parameters are fitted on the actual transfer function using the minimization of the error between the Bode plots as described in section \ref{sec:M&M}. First, the break frequency is determined from the real transfer function based on the crossing point of the magnitude plot with the $-3{dB}$ line. One decade before and after the break frequency is used to fit the models.

\subsection{modeling of higher-order processes}

The performance of the proposed model is investigated by comparing it to the state-of-the-art models found in literature.
\begin{enumerate} 
\item First Order Plus Dead Time (FOPDT)\\ $H_{1}(s) = \frac{K}{\tau s + 1}\exp(-Ls)$: the widely used standard approximation.
\item  First Order Fractional Order Plus Dead Time (FO$^2$PDT)\\ $H_{2}(s) = \frac{K}{\tau s^\alpha + 1}\exp(-Ls)$ with $\alpha\in ]0, 2[$: currently the state of the art, using a fractional pole. 
\item First Order Plus Fractional Diffusive Delay (FOPFDD)\\ $H_{3}(s) = \frac{K}{\tau s + 1}\exp(-(Ls)^{\alpha})$ with $\alpha \in ]0, 1[$: the proposed model.
\end{enumerate}
\begin{table}[h]
	\centering
	\caption{The parameters found by the optimization algorithm for the FOPDT $H_{1}(s)$, the FO$^2$PDT $H_{2}(s)$, and the FOPFDD $H_{3}(s)$ models for $n$ RC-networks in series.}
	\label{tab:OptimizedParam}
	\begin{adjustbox}{width=\textwidth}
		\begin{tabular}{llll||llll||llll}
			\hline
			\textbf{n} 	& \multicolumn{3}{c||}{$\mathbf{H_1(s)}$} & \multicolumn{4}{c||}{$\mathbf{H_2(s)}$} & \multicolumn{4}{c}{$\mathbf{H_3(s)}$} \\
			& K         & $\tau$         & L         & K     & $\tau$    & L    & $\alpha$    & K     & $\tau$    & L    & $\alpha$    \\ \hline
			4          	& 0.99      & 7.67           & 1.27      & 0.99  & 8.21      & 1.16 & 1.05        & 1.00  & 5.91      & 1.51 & 0.78        \\
			5          	& 0.99      & 11.33          & 2.00      & 0.99  & 12.32     & 1.85 & 1.05        & 1.01  & 8.46      & 2.47 & 0.77        \\
			6          	& 0.99      & 15.78          & 2.88      & 0.99  & 17.40     & 2.68 & 1.05        & 1.03  & 10.58     & 3.94 & 0.73        \\
			7          	& 0.99      & 20.93          & 3.90      & 0.99  & 23.33     & 3.64 & 1.05        & 1.03  & 13.54     & 5.54 & 0.72        \\
			8          	& 0.99      & 27.00          & 5.08      & 0.99  & 30.03     & 4.76 & 1.04        & 1.04  & 16.92     & 7.39 & 0.71        \\ 
			32         	& 0.99      & 388.61         & 77.24     & 0.99  & 423.00    & 75.50& 1.02        & 1.05  & 232.65    & 118.21& 0.69        \\
			64         	& 0.99      & 1531.10        & 304.62    & 0.99  & 2033.30   & 286.98& 1.04       & 1.06  & 846.00    & 511.13& 0.67        \\
			\hline
		\end{tabular}
	\end{adjustbox}
\end{table}

The optimization problem is prone to boundaries, namely for $H_{1}(s)$ and $H_{2}(s)$, $K$ is limited to a very narrow band around $1$ as this is the static gain of the process. For $H_{3}(s)$, $K$ is not constrained as explained in section \ref{sec:Discussion}. For $H_{2}(s)$ and $H_{3}(s)$ $\alpha$ is constrained according to their respective intervals as presented before. All other parameters are not constrained. The results of the optimization problem are given in Table \ref{tab:OptimizedParam} for several values of $n$.

\begin{figure}[!h]
	\centering
	\hspace*{-1cm}
	\begin{subfigure}[b]{0.505\textwidth}
		\centering
		\includegraphics[width=\textwidth]{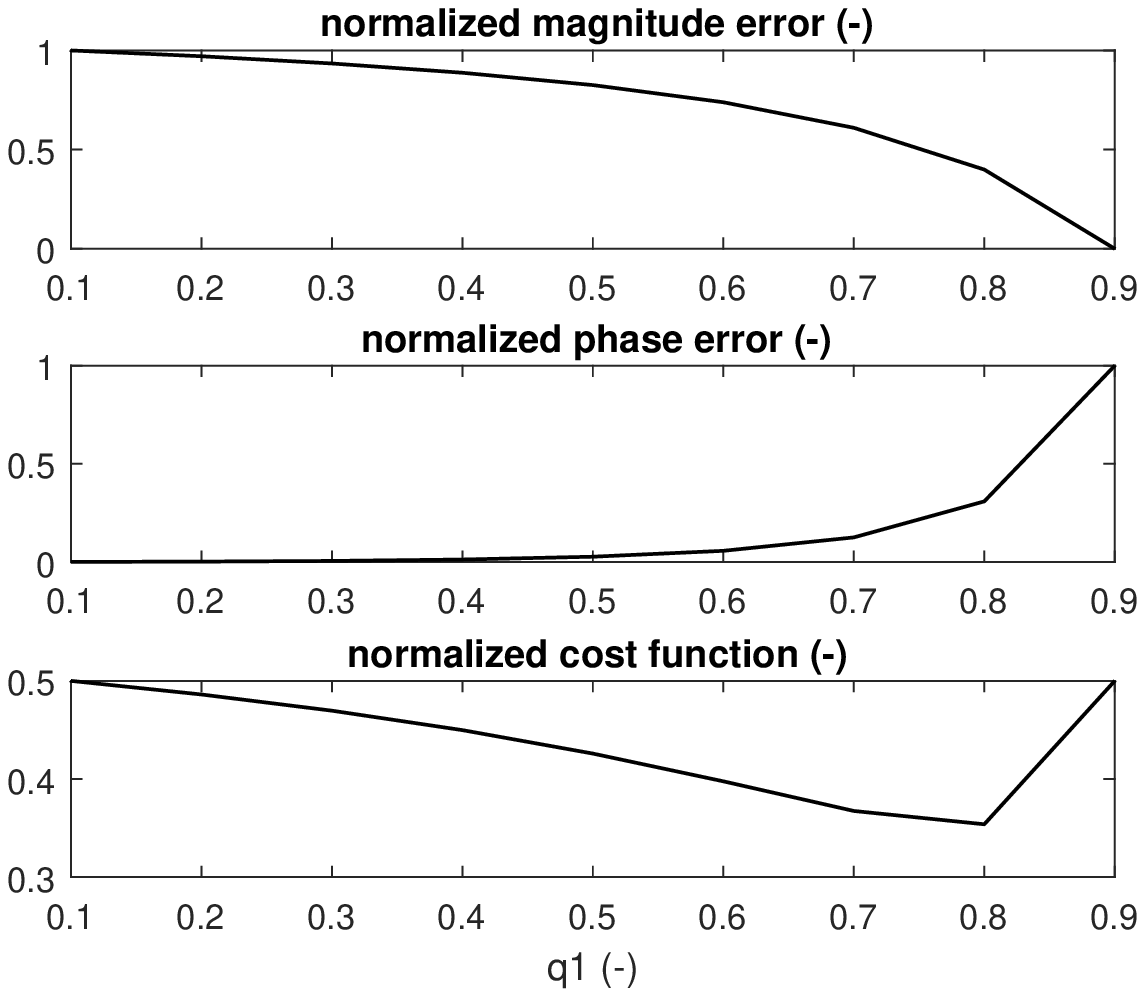}
		\caption{}
		\label{fig:n6Pareto}
	\end{subfigure}
	\hspace{0.005\textwidth}
	\begin{subfigure}[b]{0.43\textwidth}
		\centering
		\includegraphics[width=\textwidth]{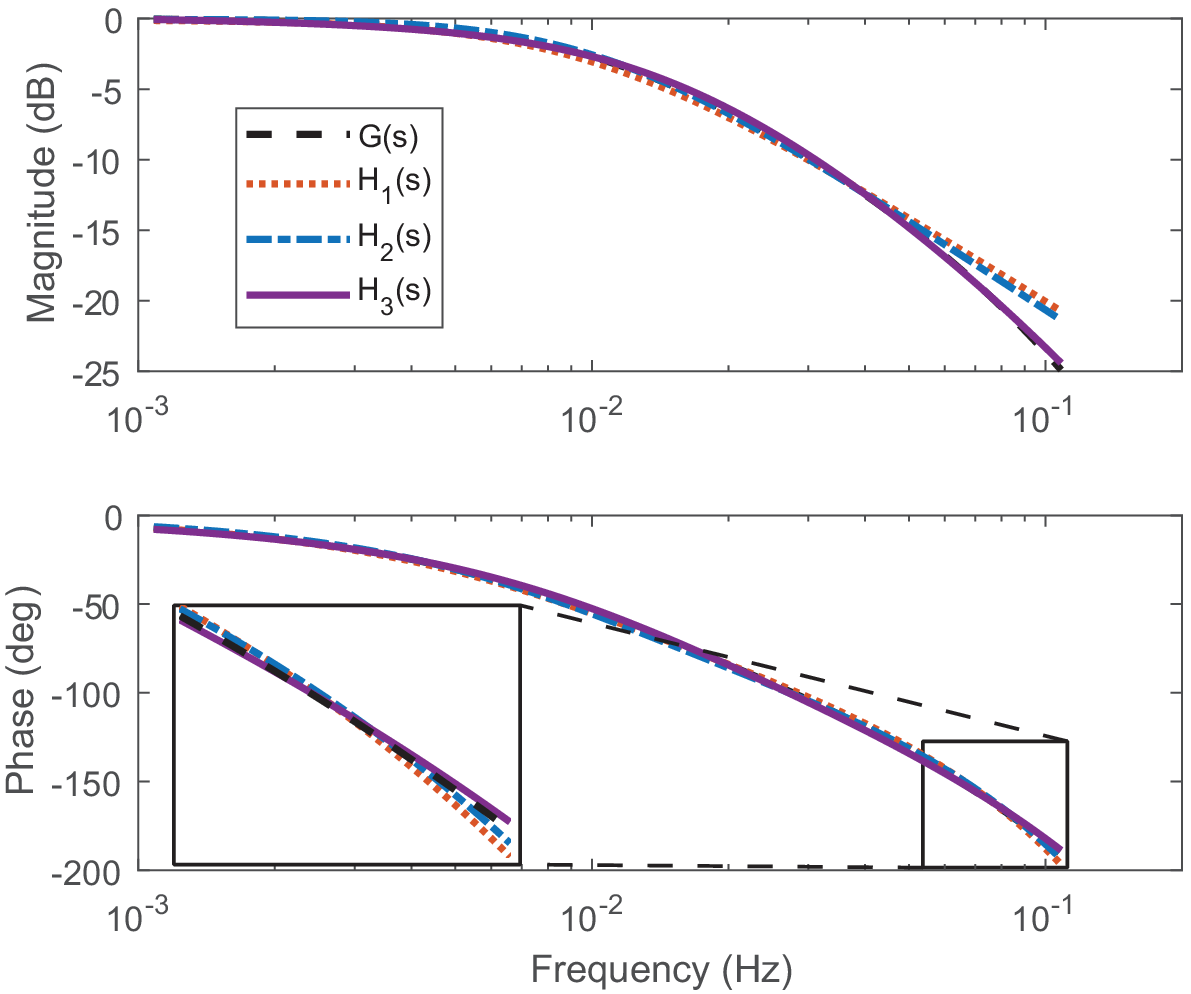}
		\caption{}
		\label{fig:n6Bode}
	\end{subfigure}
	\vskip\baselineskip
	\begin{subfigure}[b]{0.505\textwidth}
		\centering
		\includegraphics[width=\textwidth]{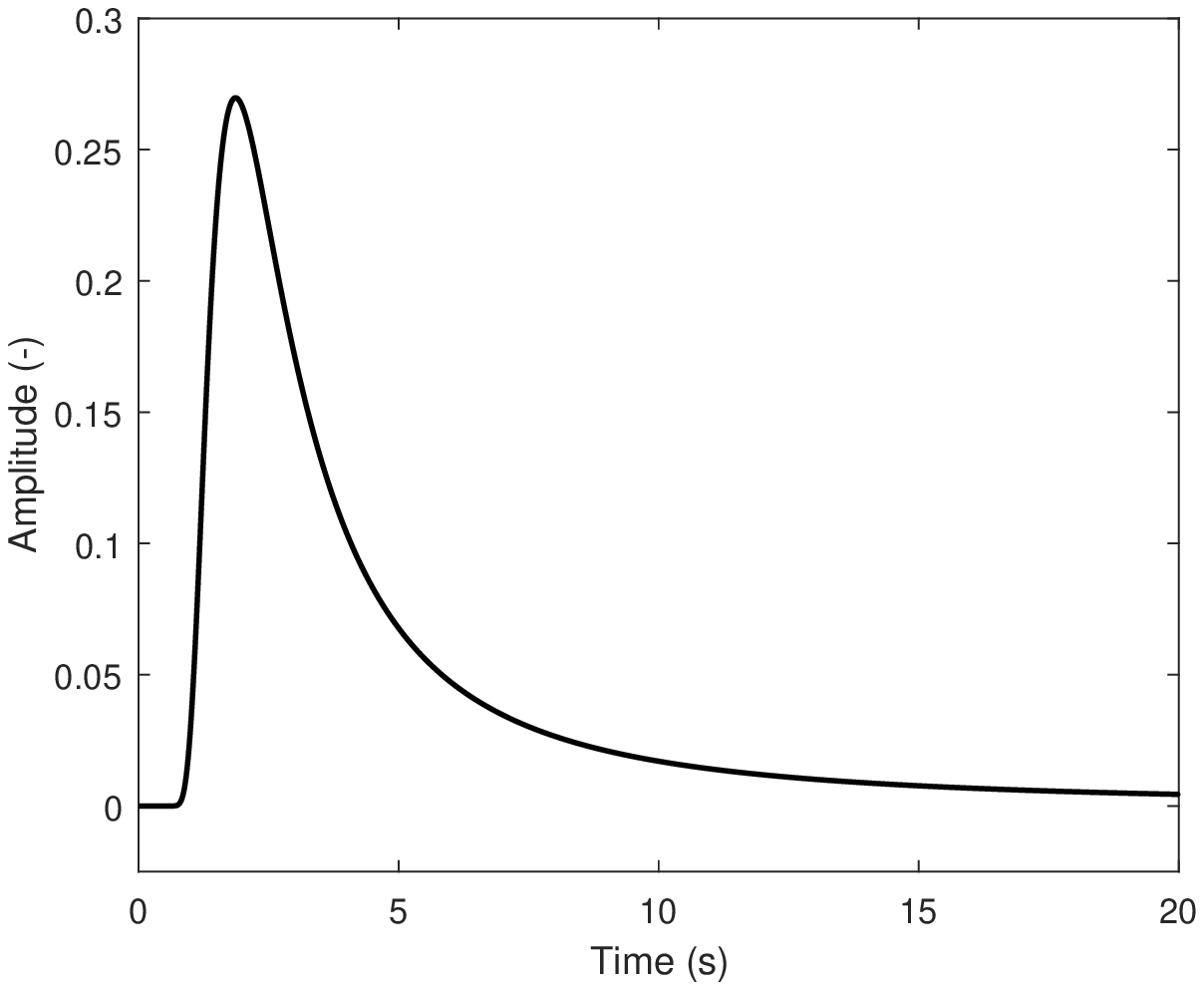}
		\caption{}
		\label{fig:n6ImpulseResp}
	\end{subfigure}
	\hfill
	\begin{subfigure}[b]{0.475\textwidth}
		\centering
		\includegraphics[width=\textwidth]{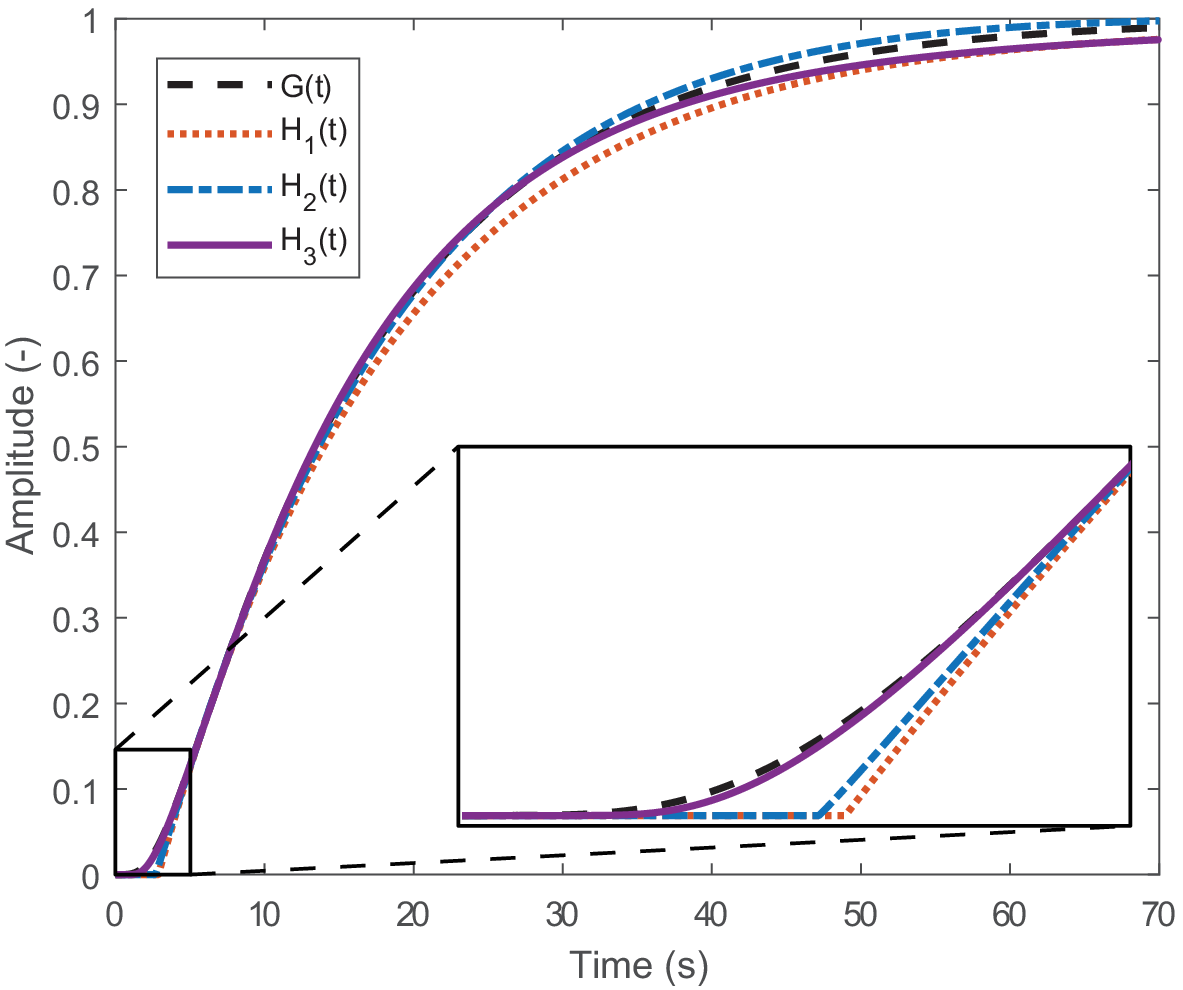}
		\caption{}
		\label{fig:n6StepResp}
	\end{subfigure}
	\caption{For $n=6$ (a) the Pareto Front from the optimization for $H_3(s)$, (b) the optimized Bode plots, (c) the impulse response of the FDD for the given $\alpha$ and $L$, and (d) the step responses are plotted.}
	\label{fig:n6}
\end{figure}

To visualize the different steps of the methodology, some plots for $n=6$ are given in Figure \ref{fig:n6}. In Figure \ref{fig:n6Pareto} the Pareto front of the optimization problem for $H_{3}(s)$ is presented to show the trade-off between minimizing the error for the magnitude plot and the phase plot respectively. The Pareto front indicates for which weight $q_{1}$ the trade-off between both objectives is optimal \cite{Juchem2019}. This leads to an optimal Bode plot for each model such that the error between the model and the real system's Bode plot is minimal for the combined objectives. These Bode plots are given in Figure \ref{fig:n6Bode}. The optimal Bode plot leads to an optimal set of parameters for each model. To evaluate the FOPFDD in time domain the FDD's impulse response is needed. Based on parameters $\{L, \alpha, t\}$ the impulse response is obtained (see Figure \ref{fig:n6ImpulseResp}) using Algorithm \ref{alg:FDD}. The step is found by calculating the convolution of the step response of the first order part of the model with the given impulse response (see \eqref{eq:convolution}). The step responses of $H_{1}(s)$ and $H_{2}(s)$ can be found using MATLAB$^{\tiny\textregistered}$ with the FOMCON toolbox \cite{Tepljakov2011}. The step responses are given in Figure \ref{fig:n6StepResp}.

To evaluate the performance of these three models a cumulative squared error between the real system and the model's step responses is calculated. To evaluate the model quality the error is represented by $J_{\eta}$ for $\eta\in [30\%, 63\%, 90\%]$. Here, $J_{\eta}$ is given by
\begin{equation}
	J_{\eta} = \frac{1}{N}\mathbf{[Y(t) - H(t)]\cdot[Y(t) - H(t)]^{T}}
\end{equation}
for $t\in [0, t_{\eta}]$ with $\eta = Y(t_{\eta})$. The vectors $\mathbf{Y(t)}, \mathbf{H(t)}\in \mathbb{R}^{1\times N}$ are the time response vectors of the real process and a model's step response respectively. In Table \ref{tab:ModelErrors} the modeling errors are given.

\begin{table}[h]
	\centering
	\caption{The cumulative squared error averaged over the samples $J_{\eta}$ for $t\in [0, t_{\eta}]$ with $\eta = G(t_{\eta})$ the percentage of the end value.}
	\label{tab:ModelErrors}
	\begin{adjustbox}{width=\textwidth}
		\begin{tabular}{llll||lll||lll}
			\hline
			\textbf{n} 	& \multicolumn{3}{c||}{$\mathbf{H_1(s)}$} & \multicolumn{3}{c||}{$\mathbf{H_2(s)}$} & \multicolumn{3}{c}{$\mathbf{H_3(s)}$} \\
						&$J_{30\%}$ & $J_{63\%}$	 & $J_{90\%}$&$J_{30\%}$ &$J_{63\%}$ 	& $J_{90\%}$&$J_{30\%}$ &$J_{63\%}$ &$J_{90\%}$  \\ \hline
			4          	& 1.41e-4   & 2.37e-4        & 5.39e-4   & 8.88e-5   & 5.67e-5 	 	& 4.83e-5 	& 1.52e-5  	& 3.22e-5   & 8.07e-5    \\
			5          	& 1.36e-4   & 2.14e-4        & 4.77e-4   & 8.90e-5   & 5.84e-5 	 	& 4.29e-5   & 9.05e-6  	& 1.16e-5   & 2.77e-5    \\
			6          	& 1.27e-4   & 2.10e-4        & 4.71e-4   & 8.25e-5   & 5.48e-5	 	& 4.15e-5   & 1.83e-6  	& 3.88e-6   & 8.41e-6    \\
			7          	& 1.23e-4   & 2.01e-4        & 4.49e-4   & 8.05e-5   & 5.37e-5 	 	& 3.93e-5   & 7.57e-7  	& 1.11e-5   & 1.62e-5    \\
			8          	& 1.24e-4   & 2.37e-4        & 5.17e-4   & 8.01e-5   & 5.30e-5 	 	& 3.62e-5   & 6.74e-7  	& 1.95e-5   & 2.71e-5    \\ 
			32         	& 1.09e-4   & 1.73e-4        & 3.87e-4   & 9.27e-5   & 6.82e-5	 	& 7.20e-5   & 3.22e-6  	& 5.39e-5   & 7.44e-5    \\
			64         	& 1.08e-4   & 1.74e-4        & 3.91e-4   & 6.96e-5   & 4.66e-5	 	& 3.22e-5   & 6.15e-6  	& 8.26e-4   & 1.05e-4    \\
			\hline
		\end{tabular}
	\end{adjustbox}
\end{table}

\section{Discussion}\label{sec:Discussion}

\setcounter{section}{5}
\setcounter{equation}{0}

\subsection{Physical meaning of the FDD}
In this work a generalized diffusion equation \eqref{eq:generalPDE} is proposed which gives rise to solution \eqref{eq:generalizedSolution} including the newly defined term the Fractional Diffusive Delay (FDD) or $\exp{\left(-(Ls)^{\alpha}\right)}$. In Figure \ref{fig:impulseRespAlpha}, the impulse responses of the FDD term with $L=1$ are given for different $\alpha$ based on a numerical simulation algorithm. From this figure, a clear agreement with existing knowledge is found. In the case of integer order $\alpha=1$ the dirac response is found. As $\alpha$ approaches 0.5, a diffusion effect can be perceived where the dirac peak is spread out in time. An effect on the delay aspect of the impulse response can also be seen in Figure \ref{fig:impulseRespAlpha}. Depending on $\alpha$, the time before the system responds significantly, varies. Also, the value of $L$ has an effect on this delay effect as can be observed in Figure \ref{fig:impulseRespL}. In conclusion, the FDD term is a function of $\alpha$ and $L$, which allows to balance between a diffusive and a delay effect.

Many processes include such a combination of diffusion and delay. Especially a series of interconnected systems give rise to this phenomenon due to the large number of poles. Some examples include transmission lines, communicating water tanks, etc. Nowadays, literature resorts to an artificial construct to represent the time delay by naively shifting the entire response in time. This leads to a discontinuity in the response, which is artificial for a physical process. However, diffusion has shown to be the missing link to explain the response's gradual build-up in time. By combining these two concepts, delay by shifting the response in time on the one hand and diffusion with its continuous and gradual response on the other hand, a time response that corresponds to a real physical response can be obtained. 

\subsection{Model fitting}
The method to minimize the error of the model's frequency response is an effective method to find the optimal set of parameters. The cost function consists of two objectives, while there are three or four model parameters. A trade-off exists and, therefore, a Pareto front is used to find the optimal weights between both objectives. In Table \ref{tab:OptimizedParam} some preliminary trends are observed. Firstly, for all models $L$ is expressed in seconds and increases as the dead time of the real system increases. This makes sense in models $H_{1}(s)$ and $H_{2}(s)$. Secondly, $\tau$ increases as well for all models with increasing $n$. For the first two models there is a trade-off between $\tau$ and $L$ to obtain a good fit on both the magnitude and the phase. In $H_{3}(s)$ this trade-off is managed mainly by $\alpha$ and partially by parameter $L$, due to the nonlinearity introduced by $\alpha$ which creates a coupling between phase and magnitude. Thirdly, for $H_{3}(s)$ parameter $K$ is given more freedom as there is a clear trade-off between modeling the low versus the high frequency behavior for this model. Nevertheless, the parameter stays close to the actual static gain, which is 1. Finally, parameter $\alpha$ is used in model $H_{2}(s)$ to have a better fit in the middle frequency range (see Table \ref{tab:ModelErrors}). But notice that only a small deviation of the integer case is needed to improve the fit. The parameters of model $H_{1}(s)$ and $H_{2}(s)$ are similar. This is also clear from Figure \ref{fig:n6Bode} for the case $n=6$. For the other cases of $n$ similar conclusions are drawn. Parameter $\alpha$ for model $H_{3}(s)$ deviates more from the integer case. There is also a clear tendency that $\alpha$ decreases for increasing $n$.

Table \ref{tab:ModelErrors} gives a cumulative error for different partitions of the step response. Three time sections are defined: A) the time needed to reach 30\% of the final value, B) the time needed to reach 63\% of the final value and C) the time needed to reach 90\% of the final value. It can be seen that the novelty of the newly presented FOPFDD model lies within partition A where the third model clearly outperforms the first and the second model. This can be understood from the detail in Figure \ref{fig:n6StepResp}. As compared to the artificial construction to model the delay in $H_{1}(s)$ and $H_{2}(s)$, which leads to a discontinuity in the derivative of the step response, the proposed model gives a gradual fit without the need to increase the number of model parameters.

Finally, the new model is able to reduce a high-order model to a model with only four parameters. It does not surprise the reader that this model introduces an improvement compared to the model $H_{1}(s)$, which only uses three parameters. However, an improvement is observed when comparing model $H_{2}(s)$ with the new model, while the number of parameters is the same. Also, a numerical solution is given to have a simple transition from frequency to time domain.

\subsection{Increasing number of subsystems}
The theory states that for an infinite number of subsystems, $\alpha$ has to be 0.5 as proven in \cite{Sierociuk2015}. However, in reality an infinite number of subsystems cannot be obtained, unless a theoretical construct such as a lumped-parameter model is used. In this paper, the hypothesis is postulated whether a finite number of subsystems, which is practically much more relevant, will lead to an $\alpha \neq 0.5$, however, no theoretical proof is provided. In Table \ref{tab:OptimizedParam}, the number of subsystems is increased to show the effect on the parameters. For the FOPFDD model, the optimization algorithm clearly shows a trend of $\alpha$ going towards $0.5$. Furthermore, as expected, the delay becomes larger with increasing $L$.

\section{Conclusion}
\setcounter{section}{6}
\setcounter{equation}{0}

This work presents the novel First Order Plus Fractional Diffusive Delay (FOPFDD) model which includes the innovative term, Fractional Diffusive Delay (FDD), expressed as $\exp{\left(-(Ls)^{\alpha}\right)}$. The novelty lies in the fact that the argument of the exponential function contains a fractional derivative $s^{\alpha}$ with $\alpha\in\mathbb{R}, \alpha\in]0,1[$. The FDD is analyzed in frequency domain and a full discussion of the inverse Laplace transform of FDD is also provided, which is innovative in the field of fractional calculus. The new FOPFDD model is tested on high-order RC circuits to indicate the advantages of this new model fitting compared to state-of-art models including both integer order and fractional order variants. This work elaborates the theoretical background of the time response of the FDD and an algorithm is developed to overcome numerical difficulties. These time domain calculations of the FDD expose a link with delay and diffusion and are unique as this has never been done before. The discussion of the results shows the added value of the FDD term in shaping time responses. The new model outperforms the state-of-art models, especially at the onset of the system's response. This property makes the FOPFDD model ideal to model delay dominant systems with an increased accuracy.

However, the authors are aware that this is only a first step in the direction of fully understanding the behavior of the Fractional Diffusive Delay. An important step is to reveil more properties of this term, such that the numerical approximation of the impulse response can be improved, especially with regard to time efficiency. 
With regard to solidifying the link between the FDD and the generalized heat diffusion equation, a proof needs to be established that if the number of discrete subsystems connected in series goes to infinity ($n\to\infty)$, means that $\alpha$ converges to 0.5. From the modeling perspective, a proper methodology to estimate the model parameters based on a measured time response would improve the model accuracy. A more profound link between the model and parameters and the behavior in time domain is crucial.

\section*{Acknowledgments}
None

\smallskip

\bibliographystyle{IEEEtran}
\bibliography{mybibfile}


\bigskip \smallskip

\it

\noindent
\hfill Received: February 2, 2021 \\[12pt]

\end{document}